%% file: main_arxiv.tex
\documentclass[11pt]{article}
\usepackage[dvipsnames]{xcolor}
\usepackage{rheaj}
\usepackage{setspace}
\usepackage{xparse}
\usepackage{soul}

\newcommand{\mypara}[1]{\medskip \noindent {\bf #1}}

\newcommand{\algemb}{\textsc{TreeEmb}}
\newcommand{\algsep}{\textsc{PruneAndSeparate}}
\newcommand{\addtree}{\textsc{UpdateTree}}
\newcommand{\alglp}{\textsc{RoundLP}}
\renewcommand{\phi}{\varphi}

\title{From Directed Steiner Tree to Directed Polymatroid Steiner Tree in Planar Graphs}

\author{Chandra Chekuri \thanks{Dept. of Computer Science, Univ. of
    Illinois, Urbana-Champaign, Urbana, IL 61801. {\tt
      chekuri@illinois.edu}. Supported in part by NSF grants CCF-1907937 and CCF-2402667.}
\and
Rhea Jain\thanks{Dept. of Computer Science, Univ. of Illinois,
  Urbana-Champaign, Urbana, IL 61801. {\tt rheaj3@illinois.edu}. 
  Supported in part by NSF grants CCF-1907937 and CCF-2402667.}
\and 
Shubhang Kulkarni\thanks{Dept. of Computer Science, Univ. of Illinois,
  Urbana-Champaign, Urbana, IL 61801. {\tt smkulka2@illinois.edu}.}
\and 
Da Wei Zheng\thanks{Dept. of Computer Science, Univ. of Illinois,
  Urbana-Champaign, Urbana, IL 61801. {\tt dwzheng2@illinois.edu}.}
\and
Weihao Zhu\thanks{Dept. of Computer Science, Univ. of Illinois,
  Urbana-Champaign, Urbana, IL 61801. {\tt weihaoz3@illinois.edu}.}
}

\begin{document}

\maketitle

\input{abstract}

\newpage
\setcounter{page}{1}

\input{intro}

\input{planar_dir_tree_emb}

\input{embedding_apps}

\input{planar_dst_lp}

\input{multiroot}

\bibliographystyle{plainurl}
\bibliography{references}
\appendix
\input{appendix_planar_embedding}
\end{document}

%% file: abstract.tex
\begin{abstract}
 In the Directed Steiner Tree (DST) problem the input is a directed edge-weighted graph $G=(V,E)$, a root vertex $r$ and a set $S \subseteq V$ of $k$ terminals. The goal is to find a min-cost subgraph that connects $r$ to each of the terminals. DST admits an $O(\log^2 k/\log \log k)$-approximation in \emph{quasi-polynomial} time \cite{GrandoniLL22,GhugeN22}, and an $O(k^{\eps})$-approximation for any fixed $\eps > 0$ in polynomial-time \cite{Zelikovsky97,Charikaretal99}. Resolving the existence of a polynomial-time poly-logarithmic approximation is a major open problem in approximation algorithms. In a recent work, Friggstad and Mousavi \cite{FriggstadM23} obtained a simple and elegant polynomial-time $O(\log k)$-approximation for DST in \emph{planar} digraphs via Thorup's shortest path separator theorem \cite{Thorup04}. We build on their work and obtain several new results on DST and related problems.
 \begin{itemize}
     \item We develop a tree embedding technique for rooted problems in planar digraphs via an interpretation of the recursion in \cite{FriggstadM23}. Using this we obtain polynomial-time poly-logarithmic approximations  for Group Steiner Tree \cite{GargKR00}, Covering Steiner Tree \cite{KonjevodRS02} and the Polymatroid Steiner Tree \cite{Calinescu_Zelikovsky_2005} problems in planar digraphs. All these problems are hard to approximate to within a factor of $\Omega(\log^2 n/\log \log n)$ even in trees \cite{HalperinK03,GrandoniLL22}. 
     \item We prove that the natural cut-based LP relaxation for DST has an integrality gap of $O(\log^2 k)$ in planar digraphs. This is in contrast to general graphs where the integrality gap of this LP is known to be $\Omega(\sqrt{k})$  \cite{ZosinK02} and $\Omega(n^{\delta})$ for some fixed $\delta > 0$ \cite{LiL22}. 
    \item We combine the preceding results with density based arguments to obtain poly-logarithmic approximations for the multi-rooted versions of the problems in planar digraphs. For DST our result improves the $O(R + \log k)$ approximation of \cite{FriggstadM23} when $R= \omega(\log^2 k)$.
 \end{itemize}
\end{abstract}

%% file: intro.tex
\section{Introduction}
\label{sec:intro}
We consider several rooted network design problems in \emph{directed} graphs and
develop new approximation algorithms and integrality gap results for them in 
planar digraphs. It is well-known that many problems in directed graphs are 
harder to approximate than their corresponding undirected graph versions. 
A canonical example, and the motivating problem for this paper, is
the Steiner Tree problem. The input is 
an undirected graph $G = (V, E)$ with edge costs
$c: E \to \R_{\geq 0}$, a root $r \in V$, and a set of terminals 
$S \subseteq V \setminus \{r\}$. The goal is to find 
a minimum cost subgraph of $G$ in which each terminal is connected to the root. 
Steiner Tree is NP-Hard and APX-Hard to approximate. There is a long and rich
history on approximation algorithms for this problem. The current best 
approximation ratio is  $\ln 4 + \eps$ \cite{ByrkaGRS13,GoemansORZ12}, 
and it is known that there is no approximation factor better than 
$\frac{96}{95}$ unless $\PTIME = \NP$ \cite{ChlebikC08}. 
Steiner Tree admits a PTAS in planar graphs \cite{BorradaileKM09}.
In this paper we consider the directed version of this problem. 
Given a directed graph $G=(V,E)$ and a vertex $r$
we use the term $r$-tree to denote a subgraph of $G$ that is a directed 
out-tree rooted at $r$; note that
all vertices in the $r$-tree are reachable from $r$ in $G$.

\mypara{Directed Steiner Tree (DST):} The input is a directed graph $G=(V,E)$ 
with non-negative edge costs $c(e)$, a root $r \in V$, and a set of 
\emph{terminals} $S \subseteq V \setminus \{r\}$. 
The goal is to find a min-cost $r$-tree that contains each terminal. 
We let $k := |S|$.

DST is a natural and fundamental network design problem. Its approximability 
has been a fascinating open problem. An easy observation shows
that DST generalizes Set Cover and hence does not admit a better than 
$(1-\eps)\log k$ approximation \cite{Feige98}.
Via a more sophisticated reduction, it is known to be hard to approximate to an 
$\Omega(\log^2 k/\log \log k)$-factor under plausible complexity assumptions 
\cite{GrandoniLL22}, and to slightly weaker $\Omega(\log^{2-\eps} k)$-factor 
unless $NP$ is contained in randomized quasi-poly time \cite{HalperinK03}.
There is a quasi-polynomial time $O(\log^2 k/\log \log k)$-approximation 
\cite{GrandoniLL22,GhugeN22,Charikaretal99}, and a 
polynomial time $O(k^{\eps})$-approximation for any $\eps > 0$ \cite{Zelikovsky97}. 
These results suggest that DST may admit a polynomial-time poly-logarithmic 
approximation. However, this has not been resolved
despite the first quasi-polynomial-time poly-logarithmic approximation
being described in 1997 \cite{Charikaretal99}. One reason is that the natural 
LP relaxation has been shown to have a polynomial-factor integrality gap of 
$\Omega(\sqrt{k})$ \cite{ZosinK02}, and more recently $\Omega(n^{\delta})$ 
for some fixed $\delta > 0$ \cite{LiL22}. 

In a recent work, Friggstad and Mousavi \cite{FriggstadM23} considered DST in 
\emph{planar} digraphs. They give a surprisingly simple and elegant algorithm 
which yields an $O(\log k)$ approximation in polynomial time.
Their algorithm is based on a divide-and-conquer approach building on Thorup's 
shortest path planar separator theorem \cite{Thorup04}.
Planar graphs are an important and useful class of graphs from a theoretical and 
practical point of view, and moreover several results on planar graphs have 
been extended with additional ideas to the larger class of minor-free families 
of graphs. Inspired by \cite{FriggstadM23}, we address approximation
algorithms in planar digraphs for several rooted network design problems that 
are closely related to DST. We formally define the problems below and
then discuss their relationship to DST. In all problems below, the input is a 
directed graph $G=((V,E), c)$, where $c: E \to \R_{\geq 0}$ denote edge costs, 
and a root $r$; the goal is to find a min-cost subgraph to satisfy some 
connectivity property from the root. 

\mypara{Directed Group Steiner Tree (DGST):} The input consists of 
$G=((V,E), c)$, $r$, and $k$ \emph{groups} of terminals 
$g_1, \dots, g_k \subseteq V \setminus \{r\}$. The goal is to find a minimum 
cost $r$-tree that contains a terminal from each group $g_i$. 

\mypara{Directed Covering Steiner Tree (DCST):} This is a generalization of 
DGST in which each of the groups $g_1, \dots, g_q \subseteq V \setminus \{r\}$ 
has an integer requirement $h_i \ge 1$, $i \in [q]$. The goal is to find a 
minimum cost $r$-tree that contains at least $h_i$ distinct terminals from each 
$g_i$.

\mypara{Directed Polymatroid Steiner Tree (DPST):} DPST generalizes the
aforementioned problems. In addition to $G$ and $r$, the input consists of
an integer valued normalized monotone submodular function (polymatroid) 
$f:2^V \rightarrow \mathbb{Z}_{\geq 0}$ (see Section~\ref{sec:prelim} for a 
formal definition). The goal is to find a minimum cost $r$-tree $T=(V_T,E_T)$ 
such that $f(V_T) = f(V)$. 

\medskip
It is not difficult to see that 
$\text{DST} \preceq \text{DGST} \preceq \text{DCST} \preceq \text{DPST}$
where we use $X \preceq Y$ to indicate that $X$ is a special case of $Y$. 
In general directed graphs it is also easy to see that DST and DGST are 
equivalent, though this reduction does not hold in planar 
graphs\footnote{Demaine et al.\ \cite{DemaineHK14} define planar
group Steiner tree in a restricted way where the groups correspond to the nodes 
of distinct faces of an embedded planar graph. There is a PTAS for this special 
case in undirected graphs \cite{BateniDHM16}, and in fact it is equivalent to 
DST in planar graphs.
However we only restrict the graph to be planar, and not the groups.}. Further, 
the known approximation ratios (and the main recursive greedy technique) for DST 
generalize to these problems \cite{GhugeN22,Calinescu_Zelikovsky_2005}.
In contrast, the situation is quite different in undirected graphs.
The undirected version of these problems,
namely Group Steiner Tree (GST) \cite{GargKR00}, Covering Steiner Tree (CST)
\cite{KonjevodRS02,GuptaS06} and Polymatroid Steiner Tree (PST) 
\cite{Calinescu_Zelikovsky_2005} have
been well-studied, and poly-logarithmic approximation ratios are known. 
We defer a detailed
discussion of the motivations and results on these problems, but we highlight 
one important connection. The known hardness of approximation for DST that we
mentioned earlier is due to the fact that it holds for the special case of 
GST in trees! Thus, the group covering requirement makes the problem(s) 
substantially harder even in undirected graphs where
Steiner tree has a simple constant factor approximation. We point out that the 
$O(\log k)$ approximation of \cite{FriggstadM23}  separates the approximability 
of DST and DGST in planar graphs since the latter is hard
to a factor of $\Omega(\log^2 k/\log \log k)$ in trees. The positive algorithmic 
result in \cite{FriggstadM23} naturally motivates the following questions.

\begin{itemize}
  \item \emph{Are there polynomial-time poly-logarithmic approximation 
  algorithms for DGST, DCST, and DPST in planar digraphs}?
  \item \emph{Is the integrality gap of the natural LP for DST and DGST and 
  DCST in planar digraphs poly-logarithmic}
  \footnote{It is not straightforward to formulate a relaxation for DPSP. 
  The other problems have known LP relaxations.}? 
\end{itemize}

\subsection{Results}
We provide affirmative answers to the first question and for part of the second 
question. Before stating our main results we set up some notation. In the 
setting of DPST we let $S = \{v \mid f(v) > 0\}$ denote the set of terminals
and let $N = |S|$. We also let $k = f(V)$.
Note that in the case of DST, $N = k$, while in the setting of DGST and DCST,
$S = \bigcup_i g_i$ and $k$ is the sum of the requirements.

We obtain poly-logarithmic approximation ratios for DGST, DCST and DPST in 
planar digraphs. These are the first non-trivial polynomial-time approximations 
for these problems, and we note that the ratios essentially match
the known approximation ratios for these problems in \emph{undirected} 
planar graphs.
\begin{theorem}\label{thm:polymatroid-main}
  For any fixed $\eps > 0$, there exists a polynomial time 
  $O\left(\frac{\log^{1+\epsilon} n \log k \log N}{\epsilon\log\log n}\right)$-approximation 
  algorithm for the Directed Polymatroid Steiner Tree on planar graphs. 
  In the special cases of Directed Group Steiner Tree and Directed Covering 
  Steiner Tree on planar graphs, the  approximation ratios can be improved to 
  $O(\log k \log^2 N)$. 
\end{theorem}

Our second result is on the integrality gap of a natural cut/flow based LP for 
DST; see \ref{DST-LP} for a formal description. In contrast to a polynomial-factor 
lower bound on the gap in general directed graphs, 
we show the following via a constructive argument.

\begin{theorem}\label{thm:integrality-gap}
  The integrality gap of (\ref{DST-LP}) is upper bounded by $O(\log^2 k)$ 
  in planar digraphs.
\end{theorem}

The bound we prove is weaker than the known $O(\log k)$ approximation (in fact
the proof is inspired by the same technique), and is unlikely to be tight. 
However, no previous upper bound was known prior to our work;
positive results have been obtained only for quasi-bipartite instances via 
the primal-dual method \cite{FriggstadKS16,FriggstadM21}.
LP based algorithms provide
several easy and powerful extensions to other problems, and are of much interest. 
The integrality gap of DST is also of interest in understanding the power and limitations
of routing vs coding in network information theory --- we refer the reader to 
\cite{AgarwalC04} and surveys on network coding \cite{FragouliS07,FragouliS08}.
We believe that the integrality gap of the natural LP for DGST and DCST is poly-logarithmic
in planar digraphs, however, there are some technical challenges in extending our 
approach and we leave it for future work.

\mypara{Multi-root versions:} Friggstad and Mousavi \cite{FriggstadM23} also considered
the multi-root version of DST and one can extend each of the problems we consider to 
this more general setting.
The input consists of multiple roots $r_1, \dots, r_R$. The goal is to find a minimum 
cost subgraph in which the relevant set of terminals  is reachable from 
\emph{at least one} of the roots.
Note that multi-root versions arise naturally
in some problems including information transmission (see the aformentioned work on 
network coding).
In general digraphs it is trivial to reduce the multi-rooted version to the single 
root version by adding an auxiliary
root vertex, but this reduction does not preserve planarity. Friggstad and Mousavi 
\cite{FriggstadM23} described an 
$O(R + \log k)$-approximation  for the multi-rooted version of DST. 
Using \emph{density}-based arguments (see Section \ref{sec:prelim}) 
combined with the aforementioned results, we obtain 
polylogarithmic approximation ratios for multi-rooted versions of all the considered problems 
in planar digraphs. For DST, our bound is better than the one in \cite{FriggstadM23} when
$R$ is $\omega(\log^2 k)$. 

\begin{theorem}\label{thm:multiroot}
  There is an $O(\log^2 k)$-approximation for the multi-rooted version of DST in 
  planar graphs. For the multi-rooted versions of DGST, DCST 
  there is a polynomial-time $O(\log k \log^2 N)$-approximation, and for DPST a 
  polynomial-time 
  $O\left(\frac{\log^{1+\epsilon} n \log k \log N}{\epsilon\log\log n}\right)$-approximation.
\end{theorem}

We note that in DGST, DCST, and DPST, the approximation factors for the multi-root 
versions actually match those of the single root setting. 

\begin{remark}\label{rem:ext}
  It is not difficult to see that the algorithm of Friggstad and Mousavi 
  \cite{FriggstadM23} and ours extends to several other rooted problems involving 
  budget constraints on cost or terminals, and prize-collecting versions;
  this is briefly discussed in Section \ref{sec:multi-root}.
\end{remark}

\begin{remark}
 \cite{FriggstadM23} observed that their approach extends to the node-weighted case. 
 The standard transformation from edge-weights to node-weights does not necessarily 
 preserve planarity, and hence the extension holds due to the specific technique. 
 Our results also hold for node weights. Even in \emph{undirected} graphs 
 there is no known polynomial-time poly-logarithmic approximation 
 for node-weighted GST --- this is because metric tree embeddings do not apply to 
 reduce the problem to trees.
 Thus, our results are new even for node-weighted undirected planar graphs.
\end{remark}

\subsection{Overview of Ideas} 
The $O(\log k)$-approximation for DST on planar graphs given by Friggstad and Mousavi 
\cite{FriggstadM23} uses a recursive divide-and-conquer structure. 
We provide a brief overview.
The algorithm uses Thorup's shortest path separator theorem applied to directed graphs:
\begin{lemma}[\cite{FriggstadM23,Thorup04}]\label{lem:dir_sep}
  Let $G$ be a planar directed graph with non-negative edge costs $c(e)$,
  non-negative vertex weights $w(v)$, and a root $r \in V$ such that every vertex 
  in $V$ is reachable from $r$. There exists a polynomial time algorithm 
  to find three shortest dipaths $P_1, P_2, P_3$ starting at $r$ such 
  that every weakly connected component of 
  $G \setminus (P_1 \cup P_2 \cup P_3)$ has at most half the vertex weight of $G$.
\end{lemma}

The high-level idea in \cite{FriggstadM23} is simple. Suppose we can guess the optimum 
solution value for a given DST instance, say $\opt$. Then one can remove all vertices 
$v$ farther than $\opt$ from $r$ (since they will not be in any optimum solution), 
and use the preceding theorem to find 3 paths of cost at most $3\opt$ such that 
removing the paths yields components, each of which contains at most half the 
original terminals. We can shrink the paths into $r$ and recurse on the "independent" 
sub-instances induced by the terminals in each component. The recursion depth is 
$O(\log k)$ which bounds the total cost to $O(\log k) \cdot \opt$.
The main issue is to implement the guess of $\opt$ in each recursive call. 
The authors obtain a quasi-polynomial time algorithm by 
brute force guessing $\opt$ to within a factor of $2$.
They obtain a polynomial-time algorithm by a refined argument where they folding the 
guessing into the recursion itself. We take an alternate perspective on this algorithm 
by \emph{explicitly} constructing the underlying 
recursion tree of the algorithm. Theorem \ref{thm:planar_dir_tree_emb} shows that we 
can view this recursion tree as a ``tree embedding'' for directed planar graphs 
that is suitable for rooted problems. The power of the embedding is that it essentially 
reduces the planar graph problem to a problem on trees which we
know how to solve. A caveat of our tree embedding is that it 
creates \emph{copies} of terminals.
Interestingly, for DGST and DPST this duplication
does not cause any issues since the definitions of these problems are rich enough to 
accommodate copies. For DCST one needs a bit more care to obtain a better bound than 
reducing it to DPST, and we describe the details in the technical section. 
This parallels the situation in undirected graphs where
probabilistic metric tree embeddings \cite{FRT} are used to reduce 
the GST, CST, and PST problems to trees, 
and furthermore, is the only known method to solve those problems. 
The formal description of the tree embedding is given below.

\begin{theorem}
\label{thm:planar_dir_tree_emb}
Let $G = (V, E)$ be a directed planar graph with edge costs $c: E \to \R_{\geq 0}$, 
a root $r \in V$, and a set of terminals $S \subseteq V$. 
Let $\gamma \leq c(E)$, and let $n := |V|$ and $k := |S|$. There exists an efficient 
algorithm that outputs a directed rooted out-tree $\calT = (V_T, E_T)$ with root $r_T$,
edge costs $c_T: E_T \to \R_{\geq 0}$ and a mapping $M: S \to 2^{V_T}$ that maps each 
terminal in $G$ to a set of terminals in $V_T$, that satisfies the following properties:
\begin{enumerate}
  \item \textbf{Size}: $|V_T| = O(k^3 \gamma)$, and for each terminal $t \in S$, 
  $|M(t)| = O(k\gamma)$. Furthermore, all $M(t)$ are disjoint from each other. \label{thm:emb:size}
  \item \textbf{Height:} The height of $T$ is at most $O(\log k)$. \label{thm:emb:height}
  \item \textbf{Projection from Graph:} For any $r$-tree $G' \subseteq G$
  with $c(G') \leq \gamma$
  there exists a $r_T$-tree $T' \subseteq \calT$ with 
  $c_T(T') = O(\log k) c(G')$, in which for each terminal $t \in S \cap G'$,
  $M(t) \cap T' \neq \emptyset$.\label{thm:emb:graph-to-tree}
  \item \textbf{Projection to Graph:} For any $r_T$-tree $T' \subseteq \calT$, 
  there exists a $r$-tree $G' \subseteq G$ with $c(G') \leq c_T(T')$ 
  and for each terminal 
  $t \in S$, if $M(t) \cap T' \neq \emptyset$ then $t \in G'$. Furthermore, we 
  can compute $G'$ efficiently. \label{thm:emb:tree-to-graph}
\end{enumerate} 
\end{theorem} 

\smallskip
Our proof of Theorem~\ref{thm:integrality-gap} on the LP integrality gap is inspired 
by the algorithm of \cite{FriggstadM23}. Instead of guessing $\opt$ we use the LP 
optimum value as the estimate. This is a natural idea, however, in order to prove an 
integrality gap we need to work with the original LP solution for the recursive 
sub-instances. We use a relatively simple trick for this wherein we overpay for the 
top level of the recursion to construct feasible LP solutions for the sub-instances 
from the original LP solution; the over payment helps us to
argue that the cost of the LP solutions for the sub-instances is only slightly larger 
and this can be absorbed in the recursion since
the problem size goes down. 

\smallskip
Finally, for the multi-rooted version we rely on a simple reduction to the single root 
problem via the notion of density, which is a standard idea in covering problems.

\input{related_work}

\input{prelim}

%% file: related_work.tex
\subsection{More on related work}
\label{sec:relatedwork}
There is extensive literature on algorithms for network design in both undirected and
directed networks with more literature on undirected network design.
Standard books on combinatorial optimization \cite{Schrijver-book,Frank-book},
and approximation algorithms \cite{Vazirani-approx,WilliamsonS-approx}
cover many of the classical problems and results. We also point to 
the surveys \cite{GuptaK11,KortsarzN10} on network design.
In this section we describe some closely related work and ideas. 

\mypara{Directed Steiner Tree:} Zelikovsky \cite{Zelikovsky97} was the first to address 
the approximability of DST. He obtained an $O(k^{\eps})$-approximation for any fixed 
$\eps > 0$ via two ideas. He defined a recursive greedy algorithm and analyzed its 
performance as a function of the depth of the recursion. He then showed that one can 
reduced the problem on a general directed graph to a problem on a depth/height $d$ DAG 
(via the transitive closure of the original graph) at the loss of an approximation 
factor that depends on $d$. Charikar et al \cite{Charikaretal99} refined the algorithm 
and analysis in \cite{Zelikovsky97} and combined it with the depth reduction, they 
showed that one can obtain an $O(d^2 k^{1/d}\log k)$ approximation in $O(n^{O(d)})$-time;
this led to an $O(\log^3 k)$-approximation in quasi-polynomial time. 
Subsequentaly Grandoni et al \cite{GrandoniLL22}
improved the approximation to $O(\log^2 k/\log \log k)$ in quasi-polynomial time 
via a more sophisticated LP-based approach. A different approach that also yields the same bound was given by Ghuge and Nagarajan \cite{GhugeN22}
and this is based on a refinement of the recursive greedy algorithm for walks in graphs \cite{ChekuriP05}. The advantage
of \cite{GhugeN22} is that it yields an $\Omega(\log \log k/\log k)$-approximation
in quasi-polynomial time algorithm for the budgeted
version of DST; the goal is to maximize the number of terminals in a $r$-rooted tree with a given budget of
$B$ on the cost of the tree. 

Zosin and Khuller \cite{ZosinK02} showed that the natural cut-based LP relaxation has an integrality gap of
$\Omega(\sqrt{k})$ for DST. However, their example only showed a gap of $\Omega(\log n)$ as a function of
the number of nodes $n$. There was some hope that the integrality gap is poly-logarithmic in $n$, however
\cite{LiL22} recently showed that the gap is $\Omega(n^{\delta})$ for some $\delta > 0$ by modifiying
the construction in \cite{ZosinK02}. Interestingly these lower bound examples are DAGs with $O(1)$-layers
for which the recursive-greedy algorithm yields an $O(\log k)$-approximation in polynomial-time!
Rothvoss \cite{Rothvoss11} showed that $O(\ell)$-levels of the Lasserre SDP hierarchy when applied to the standard cut-based LP reduces the integrality gap to $O(\ell \log k)$ on DAGs with $\ell$ layers.
This was later refined to show that $O(\ell)$-levels of
the Sherali-Adams hierarchy suffices \cite{FriggstadKKLST14}. However, both these approaches also require quasi-polynomial time to obtain a poly-logarithmic approximation. 

\mypara{Group Steiner Tree:} The group Steiner tree problem (GST) in undirected graphs was
introduced by Reich and Widemeyer \cite{ReichW89} and it was initially motivated by an application in VLSI design.
Garg, Konjevod and Ravi \cite{GargKR00} obtained an $O(d \log k)$ approximation in depth $d$ trees via an elegant randomized
rounding algorithm of the fractional solution to a natural LP relaxation; one can reduce the depth to $O(\log N)$ via the
fractional solution and hence they obtained an $O(\log N \log k)$-approximation. They obtained an algorithm for general
graphs via probabilistic tree embeddings \cite{FRT}. Zosin and Khuller obtained an alternate deterministic $O(d \log k)$-approximation on trees.
The randomized algorithm of \cite{GargKR00} can also be derandomized via standard methods \cite{Charikar_Chekuri_Goel_Guha_1998}. 
The integrality gap of the natural LP for GST was shown to be $\Omega(\log^2 k/(\log \log k)^2)$ by Halperin et al. \cite{Halperinetal07}.
This gap motivated the inapproximability result of Halperin and Krauthgamer who showed that GST in trees is
hard to  approximate within a factor of $\Omega(\log^{2-\eps} k)$. This was further improved to $\Omega(\log^2 k/\log \log k)$ \cite{GrandoniLL22}
under stronger complexity theoretic assumption.
Note that one can consider node-weighted GST. In general directed graphs one can see that
node-weighted and edge-weighted problems are typically reducible to each other, however this is not necessarily the case in undirected graphs.
The known approaches to approximate GST in general graphs in polynomial time uses probabilistic tree
embeddings (or, more recently, oblivious routing trees \cite{Racke08,CGL15,cllz22} which are intimately connected to tree embeddings). However, node-weighted
problems do not admit such tree embeddings and thus we do not have polynomial-time poly-logarithmic approximation for GST in
node-weighted undirected graphs. 

There is a strong connection between GST, its directed counterpart DGST and DST.
As we remarked, it is easy to see that in directed graphs, DST and DGST are equivalent. One can also reduce
GST to DST by adding a dummy terminal $t_i$ for each group $g_i$ and connecting all the vertices in $g_i$
to $t_i$ via directed edges. Thus GST admits an $O(\log^2 k/\log \log k)$-approximation in quasi-polynomial time 
though the best polynomial-time algorithm loses another log factor due to tree embeddings.
On the other hand, via the height reduction approach and path expansion, one can reduce DST to GST in trees in quasi-polynomial-time
at the loss of an $O(\log k)$ in the approximation ratio (details of this
are essentially folklore but can be seen in \cite{ChekuriEGS11}). This partially explains the reason why the hardness results for DST are essentially
based on the hardness of GST in trees.

\mypara{Covering Steiner Tree and Polymatroid Steiner Tree:}
The Covering Steiner Tree problem was first considered by Konjevod, Ravi and Srinivasan \cite{KonjevodRS02} as a common generalization
of GST and the $k$-MST problems. They obtained a poly-logarithmic approximation by generalizing the ideas from GST (see also \cite{EvenKS05}). 
Gupta and Srinivasan subsequently improved the ratio \cite{GuptaS06}. Calinescu and Zelikovsky \cite{Calinescu_Zelikovsky_2005} defined
the general Polymatroid Steiner Tree problem (PSP) and its directed counterpart (DPSP). They were motivated by both theoretical considerations as well as applications in wireless networks. Submodularity provides substantial power to model a variety of problems. PSP is easily seen to generalize GST and CST.
However, unlike GST and CST, even in trees there is no easy LP relaxation for PSP that one can formulate, solve and round. 
Thus, \cite{Calinescu_Zelikovsky_2005} used a different approach. Chekuri, Even and Kortsarz \cite{Chekuri_Even_Kortsarz_2006} had shown
that the recursive greedy algorithm of \cite{Charikaretal99} can be adapted to run in polynomial-time on trees after preprocessing it
to reduce the degree and height. The recursive greedy approach naturally generalizes to PSP/DPSP just as the greedy algorithm
for Set Cover generalizes to Submodular Set Cover \cite{Wolsey82}. Via this generalization, \cite{Calinescu_Zelikovsky_2005} obtained polynomial-time
approximation algorithms for PSP in trees and hence in general graphs via tree emebeddings. For DPST they obtained quasi-polynomial-time
approximation algorithms.

%% file: prelim.tex
\subsection{Preliminaries}
\label{sec:prelim}
Let $G = (V, E)$ be a directed graph with edge costs $c:E\rightarrow\R_+$. 
For $E'\subseteq E$, we denote $c(E') = \sum_{e \in E'}c(e)$. 
We assume all edge costs $c(e) \geq 1$ and are polynomially bounded in $n$. For problems 
considered in this paper, this is without loss of generality by guessing the cost of the optimal 
solution $\opt$, contracting edges with cost much smaller than $\opt$, and scaling appropriately.

We define the minimum density DST problem. 
\begin{definition}
	Given an instance of DST on a graph $G=(V,E)$ with root $r$, 
	the \emph{density} of a partial solution $F \subseteq E$ is $c(F)/k(F)$ where 
	$k(F)$ is the number of terminals in $S$ that have a path from $r$ in $G[F]$. 
	The minimum-density DST problem is to compute a solution of minimum density in a 
	given instance of DST.
\end{definition}
One can similarly define minimum density versions of DPST 
(which generalizes DGST and DCST); for a partial solution $F \subseteq E$, we let 
$S_F$ denote the set of terminals in $S$ that have a path from $r$ in $G[F]$. 
The density of $F$ is $c(F)/f(S_F)$, where $f$ is the given polymatroid.

\medskip

\mypara{Graph notation.}
We use $V(H)$ and $E(H)$ to refer to the vertices and edges of a graph $H$ 
when the vertex and edge sets have not been explicitly specified.
For $S \subseteq V$, we use $E[S]$ to denote the set of edges of $E$ with both 
endpoints in $S$, $G[S]$ to denote the subgraph $(S, E[S])$ \emph{induced by $S$} 
in $G$, and $\delta^+(S) = \{(u, v) \in E : u \in S, v \not \in S\}$ to denote 
the \emph{out-cut} of $S$. For $u \in V$ we use $\mathbbm{1}_{u\in S}$ to denote 
the \emph{indicator} of the vertex $u$ being in $S$, i.e. 
$\mathbbm{1}_{u\in S} = 1$ if $u \in S$ and is $0$ otherwise.
For a path $P \subseteq G$, we define the \emph{length} as the number of edges on the path.
For a given edge-cost function $c$, we denote $d_c(r,t)$ as the length of 
\emph{shortest $r$-$t$ path} in $G$ with edge weights $c$; we drop the subscript 
$c$ if it is clear from context.
For $r \in V$ an \emph{out-tree rooted at $r$} is a subgraph $T = (V_T, E_T) 
\subseteq G$ such that there is a unique $r$-$v$ path for every $v \in V_T$.
The \emph{height} of the tree is the maximum length of a $r$-$v$ path, for $v \in V_T$. 
The \emph{size} of the tree is the number of vertices in the tree $|V_T|$.
For a subgraph $G' \subseteq G$, we use $G/G'$ to denote the graph obtained by 
contracting every edge of $G'$ and $G - G'$ to denote the graph obtained by deleting 
every edge in $G'$.
A \emph{weakly connected component} of $G$ is a connected component of the underlying 
undirected graph obtained from G by ignoring the edges orientations. 

\mypara{Submodular functions.}
Let $f: 2^{V} \rightarrow \mathbb{R}$ be a set function over ground set $V$. 
The function $f$ is \emph{monotone} if  $f(X)\le f(Y)$ for every 
$X\subseteq Y\subseteq V$, \emph{submodular}
if $f(X)+f(Y)\ge f(X\cap Y)+f(X\cup Y)$ for every $X,Y\subseteq V$, 
and \emph{normalized} if $f(\emptyset) = 0$. An integer-valued, normalized, 
monotone and submodular function is called a \emph{polymatroid}.

%% file: planar_dir_tree_emb.tex
\section{Recursive Tree Embeddings for Directed Planar Graphs}
\label{sec:emb}

We show that we can view the recursion tree given by the algorithm of 
\cite{FriggstadM23} as a tree embedding by proving Theorem 
\ref{thm:planar_dir_tree_emb}.
As described in Section \ref{sec:intro}, the algorithm of \cite{FriggstadM23} 
starts with an upper bound $\gamma$ for the cost of an optimal solution 
(we call this $\opt$). In order to fold the guessing of $\opt$ into the recursion, 
the algorithm makes two recursive calls and takes the minimum. For the first 
recursive call, it applies Lemma \ref{lem:dir_sep} to
obtain a planar separator, buys the separator, and recurses on the resulting 
weakly connected components. 
The second recursive call divides the ``guess'' $\gamma$ by two.
The polynomial runtime comes from the fact that at each step, we either 
halve the guess of $\opt$ or halve the number of terminals.

We define a subroutine $\algsep$ (Algorithm \ref{algo:prune_sep}) to describe 
the first recursive call, which
takes as input a graph $(G = (V, E), c)$ with 
root $r \in V$, terminals $S \subseteq V$, and a guess $\gamma$ for the cost of the 
optimal solution.
$\algsep((G,c), r, S, \gamma)$ removes all vertices further than $\gamma$ away 
from $r$, and uses Lemma \ref{lem:dir_sep} on the resulting graph with vertex weights 
set to 1 on the terminals and 0 elsewhere. This yields a \emph{planar separator} 
$P := P_1 \cup P_2 \cup P_3$ in which each resulting component of $G \setminus P$ has at 
most half the terminals.  
The subroutine contracts $P$ into $r$; each component of $G \setminus P$ 
corresponds to a new subinstance induced by the terminals in that component 
along with $r$. $\algsep((G, c), r, S, \gamma)$ returns $P$ 
along with the subinstances for each component. 

\begin{algorithm}[H]
    \caption{Prune and Separate Procedure}
    \label{algo:prune_sep}
    \vspace{1mm}
  $\algsep((G = (V, E), c), r, S, \gamma):$
  \begin{algorithmic}
    \State Delete all vertices $v \in V$ with $d_c(r, v) > \gamma$.
    \State Let $P := P_1 \cup P_2 \cup P_3$ be given by applying Lemma 
    \ref{lem:dir_sep} with weights $w(v) = \mathbbm{1}_{v\in S}$ for $v \in V$.
    \State Let $G_P$ be obtained from $G$ by contracting $P$ into $r$.
    \State Let $C_1', \dots, C_\ell'$ be components of $G \setminus P$, and let 
    $C_i \gets G_P[C_i' \cup \{r\}]$ \\
    \Return $(P, C_1, \dots, C_\ell)$
  \end{algorithmic}
\end{algorithm}

Given this subroutine, the tree embedding is simple. We define a recursive function 
$\algemb$, which takes as input a graph $(G = (V, E), c)$ with 
root $r \in V$, terminals $S \subseteq V$, and a ``guess'' $\gamma$. 
The algorithm instantiates a root node $r_T$ and constructs two trees 
corresponding to the two recursive calls 
made by \cite{FriggstadM23}:
\begin{enumerate}[(1)]
    \item Call $\algsep((G, c), r, S, \gamma)$ and construct the trees 
    $\calT_i = \algemb(C_i, r, S_i, \gamma)$ recursively 
    for each subinstance $(C_i, r, S_i)$. Add an auxiliary node $v^*$
    and connect it to the root of each subtree with a zero-cost edge. \label{emb-informal-sep}
    \item Recursively construct the tree $\calT_h = \algemb((G,c), r, S, \gamma/2)$.\label{emb-informal-half} 
\end{enumerate}

$\algemb((G, c), r, S, \gamma)$ connects $r_T$ to the root of $\calT_h$ 
with a zero-cost edge. It also connects $r_T$ to 
$v^*$ with an edge of cost $c(P)$, where $P$ is the planar separator 
constructed in \ref{emb-informal-sep}. These edge costs 
from $r_T$ correspond to the costs of choosing 
each recursive path. See Figure \ref{fig:emb} for a summary. 

As described in Theorem \ref{thm:planar_dir_tree_emb}, 
we would like this tree embedding 
to maintain some representation of the terminals $S$. 
It is not immediately clear how one could accomplish this; the first 
recursive call decomposes $G$ while the second takes in a copy of $G$, so the 
same terminals can appear in both corresponding subtrees. Therefore, we need to 
allow for multiple copies of the same terminal. 
The algorithm of \cite{FriggstadM23} includes a terminal $t$ in the 
solution either when $t$ is in some planar separator or when $t$ is the only 
remaining terminal 
in $S$, in which case it buys the shortest path $r$-$t$ path. 
To represent this in the 
tree embedding, we create a copy of a terminal for every separator or base case it is 
in. We denote copies of $t$ as $t^P$ where $P$ is the separator 
or shortest path containing $t$, and let $M(t)$ denote the set of all copies of $t$.
\begin{figure}
    \centering
    \includegraphics[width=0.5\linewidth]{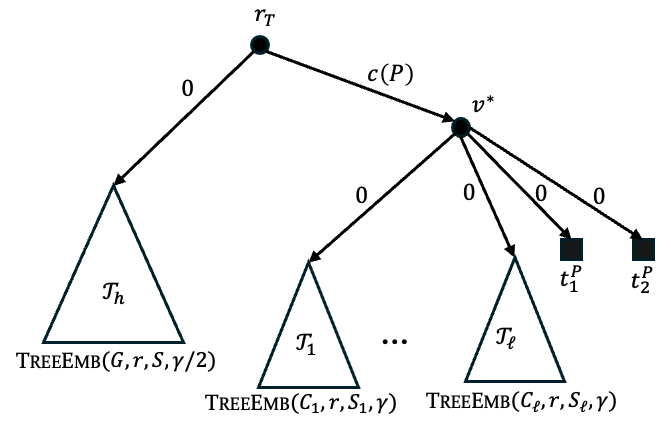}
    \caption{The tree $\calT$ given by $\algemb$, where $t_1, t_2$ are terminals in $P$.}
    \label{fig:emb}
\end{figure}

The full algorithm is described in Algorithm \ref{algo:dir_plan_emb}. 
We use a subroutine \\
$\addtree((\calT, M), (\calT', M'), v)$, which updates the tree
$\calT$ to include $\calT'$ via a zero-cost edge from $v \in V(\calT)$ 
to the root of $\calT'$. This subroutine also updates the mapping $M$ to include $M'$.
\begin{algorithm}
  \caption{Tree Embedding Construction for Directed Planar Graphs} 
  \label{algo:dir_plan_emb}
  \vspace{1mm}
  $\algemb\left(\left(G = \left(V, E\right), c\right), r, S, \gamma\right)$:
  \begin{algorithmic}
    \If{$\gamma < 1$ or $S = \emptyset$} \Return null \label{algemb:line:bc0}
    \EndIf
    \State Initialize $r_T$, $\calT \gets \left(V_T = \left\{r_T\right\}, E_T = \emptyset\right)$ as tree embedding with empty cost function $c_T$\newline
     Initialize $M(t) \gets \emptyset$ for all $t \in S$ as the terminal copies
    \If{$|S| = 1$}
        \State Let $P$ be a shortest $r$-$t$ path in $G$ 
        \State Add a new vertex $t^P$ to $V_T$ and to $M(t)$ 
        \State Add a new edge $e_t = (r_T, t^P)$ to $E_T$, and define $c_T(e_t) = c(P)$.
        \State \Return $(\calT, r_T, M)$
    \EndIf
    \vspace{1mm}
    \State Recursively construct $(\calT_h, r_h, M_h) \gets \algemb((G,c), r, S, \gamma/2)$ 
    \State $\addtree((\calT, M), (\calT_h, M_h), r_T)$
    \vspace{1mm}
    \State $(P, C_1, \dots, C_\ell) \gets \algsep((G,c), r, S, \gamma)$ 
    \State Add a new vertex $v^*$ to $V_T$, edge $e^* = (r_T, v^*)$ to $E_T$ and define $c_T(e^*) = c(P)$.
    \For{$i \in [\ell]$}
        \State Recursively construct $(\calT_i, r_i, M_i) \gets \algemb(C_i, r, S \cap C_i, \gamma)$ 
        \State $\addtree((\calT, M), (\calT_i, M_i), v^*)$
    \EndFor 
    \For{$t \in S \cap P$} \label{algemb:line:terminal-start}
        \State Add a new vertex $t^P$ to $V_T$ and $M(t)$
        \State Add a new edge $e_t = (v^*, t^P)$ to $E_T$, and define  $c_T(e_t) = 0$.
    \EndFor 
    \State\Return $(\calT, r_T, M)$
  \end{algorithmic}
  \vspace{3mm}
  $\addtree((\calT, M), (\calT', r', M'), v):$
  \begin{algorithmic}
    \State $V(\calT) \gets V(\calT) \cup V(\calT')$
    \State $E(\calT) \gets E(\calT) \cup E(\calT')$, where $c_T(e) = c_{T'}(e)$ for all $e \in E(\calT')$ 
    \State For every terminal $t \in S$, modify $M(t) \gets M(t) \cup M'(t)$
    \State $e \gets (v, r'), c_T(e) = 0$; add $e$ to $E(\calT)$\\
    \Return $((\calT = (V(\calT), E(\calT)), c_T), M)$
  \end{algorithmic} 
\end{algorithm}

We claim that the tree and mapping $(\calT, r_T, M)$ given by 
$\algemb(G, r, S, \gamma)$ satisfies all properties 
of Theorem \ref{thm:planar_dir_tree_emb}. We outline the ideas here.
Properties \ref{thm:emb:size} and \ref{thm:emb:height} 
follow easily from simple inductive arguments similar to those in \cite{FriggstadM23}. 
We note that the construction given in Algorithm \ref{algo:dir_plan_emb} 
constructs a tree of height $O(\log(k\gamma))$; this can be improved to 
$O(\log k)$ (see Remark \ref{rem:embedding_height_reduction}). 
For Property \ref{thm:emb:graph-to-tree}, let $G' \subseteq G$ be an $r$-tree.
Let $i$ be such that $\gamma/2^{i+1} \leq c(G') \leq \gamma/2^i$. 
To construct $T' \subseteq \calT$, we include the path from 
$r_T$ to the root of the subtree given by $\algemb(G, r, S, \gamma/2^i)$. 
We then include the planar separator branch, which costs at most 
$3 \gamma/2^{i} \leq 6 c(G')$ and recurse on sub-instances. 
The cost bounds and terminal copy requirements follow immediately from the 
feasibility and cost analysis of \cite{FriggstadM23}. 
For Property \ref{thm:emb:tree-to-graph},
let $T' \subseteq \calT$ be an $r_T$-tree. Notice that there are only 
two types of non-zero cost edges in $\calT$; those corresponding to planar separators 
or those corresponding to shortest paths in the base case. 
Let $\calP(T')$ be the set of all such paths and separators corresponding to 
non-zero cost edges of $T'$. We let $G' = \cup_{P \in \calP(T')} P$.
It is clear that $G'$ can be computed efficiently by traversing $T'$ and 
including all relevant paths. The cost and terminal copy guarantees are simple; 
see Appendix \ref{appendix:sec:emb} for detailed proofs.

\begin{remark}
\label{rem:embedding_height_reduction}
  The height of the tree can be reduced to $O(\log k)$ by increasing the 
  degree by a factor of $O(\log \gamma)$. Instead of only making two recursive 
  calls, one can simultaneously make recursive calls 
  $\algemb(G, r, S, \gamma/2^i)$ for $i \in [\log \gamma]$, 
  along with the recursive call using the planar separator.  
  Each recursive call then only considers the ``planar separator'' branch and 
  proceeds inductively.
\end{remark}

%% file: embedding_apps.tex
\section{Group, Covering, and Polymatroid Directed Steiner Tree} 
\label{sec:applications}
In this section we give an overview of the proof of Theorem 
\ref{thm:polymatroid-main}, providing polynomial time 
polylogarithmic approximation algorithms for DGST, DCST, and DPST. 
Although DPST generalizes DCST and DGST, we discuss each of the three 
problems separately in this section since we obtain better approximation 
ratios for DGST and DCST; moreover, our algorithmic techniques for DGST and 
DCST are different.
For each of these problems, let $G = (V, E)$ denote the input graph, 
$c:E\rightarrow\R_+$ denote the edge costs, $r \in V$ denote the root, and 
$S$ denote the set of terminals.
The embedding theorem given by Theorem \ref{thm:planar_dir_tree_emb} allows 
us to effectively reduce to special cases of the problems in which the input 
graph is a tree, as described by the following high-level framework:

\begin{enumerate}[(a)]
	\item Use Theorem \ref{thm:planar_dir_tree_emb} on inputs $(G, c), r, S$, 
	and $\gamma = c(E)$ to obtain a directed out-tree $\calT = (V_T, E_T)$ 
	rooted at $r_T$ with edge costs $c_T$, and for each terminal $t \in S$ 
	a set of ``copies'' $M(t) \subseteq V_T$. The new set of terminals $S_T$ 
	is the collection of all copies $\cup_{t \in S}M(u)$. 
	\item Compute an approximate solution to a 
	relevant problem on $\calT$, \label{app:solve-on-tree}
	\item Project the solution on $\calT$ to the graph $G$ using Property 
	\ref{thm:emb:tree-to-graph} of Theorem \ref{thm:planar_dir_tree_emb}.
\end{enumerate}
\vspace{1mm}

One challenge in directly applying the above framework is constructing the 
instance and problem to solve on $\calT$ in step (b)
above; this is because a terminal in $G$ contains several copies in $\calT$.
For DGST and DPST, the ability to deal with copies of terminals is quite 
naturally instilled into the problem definitions themselves. 
In DGST, we can simply expand the groups to include all copies of a terminal, 
and in DPST, we can appropriately redefine the underlying submodular function 
and rely on the diminishing marginal returns property.
Thus for DGST and DPST, we can directly solve the equivalent problem on the tree, as 
explained below:

\mypara{Directed Group Steiner Tree:} The input is a graph $(G, c)$ with a root 
$r$ and $k$ groups $g_1, \dots, g_k \subseteq V$. After applying step (a) 
of the framework, we consider the following instance of DGST. 
The input graph is the constructed tree $(\calT, c_T)$ with root $r_T$ and 
terminal set $S_T$. The new groups $g'_1, \ldots g'_k \subseteq V_T$ are 
defined as $g_i' := \cup_{u \in g_i}M(u)$ for every $i \in [k]$. 
To obtain an approximate solution on this instance, we use the following result 
by Zosin and Khuller \cite{ZosinK02} 
(also see \cite{GargKR00, Charikar_Chekuri_Goel_Guha_1998}).
\begin{theorem}[\cite{ZosinK02}]
\label{thm:group_steiner_tree:Zosin_Khuller}
  There exists a polynomial time $O\left(d\log k\right)$-approximation algorithm 
	for Group Steiner Tree when the input graph is a tree with height $d$.
\end{theorem}
Applying this result in conjunction with Property \ref{thm:emb:height} of 
Theorem \ref{thm:planar_dir_tree_emb} regarding the height of the tree 
$\calT$ gives us an $O(\log N \log k)$-approximation for the instance of DGST 
on $(\calT, c_T)$.

\mypara{Directed Polymatroid Steiner Tree:} The input consists of a graph 
$(G, c)$ with root $r$ and a polymatroid function 
$f:2^V \rightarrow\mathbb{Z}_{\geq 0}$. After applying step (a) of the 
framework, we consider the following instance of DPST. 
The input graph is the constructed tree $(\calT, c_T)$ with root $r_T$. 
The new polymatroid function $f_T:2^{V_T} \rightarrow\mathbb{Z}$ is defined as  
$f_T(Z) : = f\left(\left\{t \in S: M(t) \cap Z \not = \emptyset\right\}\right)$ 
for every $Z \subseteq V_T$. It is not hard to see that $f_T$ is a polymatroid, 
and that an evaluation oracle for $f$ can be used to construct an evaluation 
oracle for $f_T$ in polynomial time. To obtain an approximate solution for 
this instance, we directly apply the following result by Calinescu and 
Zelikovsky \cite{Calinescu_Zelikovsky_2005}:
\begin{theorem}[\cite{Calinescu_Zelikovsky_2005}]
\label{lem:polymatroid-trees}
For every $\epsilon > 0$,
there exists a polynomial-time 
$O\left(\frac{\log^{1+\epsilon}n\log k}{\epsilon\log\log n}\right)$-approximation 
algorithm for the Polymatroid Steiner Problem when the input graph is a tree,
assuming a polynomial time oracle for the polymatroid function.
\end{theorem}

\vspace{1mm}
In both DGST and DPST, we apply step (c) to project the given solutions back to 
a solution on the original input graphs. By Property \ref{thm:emb:graph-to-tree}, 
this ``tree-embedding'' framework loses an additional $O(\log N)$ factor in the cost.
It is simple to verify that the correctness and approximation guarantees follow from 
appropriately applying properties of Theorem \ref{thm:planar_dir_tree_emb}.

\subsection{Directed Covering Steiner Tree}
We are given a graph $(G, c)$ with root $r$, and $q$ groups 
$g_1, \dots, g_q \subseteq V$, with requirements $h_1, \ldots, h_q$ 
respectively, where $\sum_i h_i = k$.
The algorithm for DCST is more involved than simply instantiating the 
framework to solve a DCST instance on trees. Technical complications arise 
because after applying step (a) of the framework to obtain a tree embedding 
$(\calT, c_T)$ and expanding groups to include all copies of a terminal 
(while keeping the same requirements, say),
a solution for DCST on this instance could satisfy the requirement of the 
$i^{th}$ expanded group by picking multiple copies of a single terminal from 
$g_i$. Consequently, it is unclear how to map such a solution on the tree back 
to a solution of our original DCST instance.

To circumvent this issue, we use an LP-based approach on the tree 
$(\calT, c_T)$ and expanded groups $g_1', \dots, g_q'$. 
The natural flow-based LP relaxation on DCST sends a flow of $h_i$ from the 
root to each group. We aim to modify this LP to bound the amount of flow 
reaching each set of copies $M(t)$.
To that end, we define an LP with flow variables $f_v$ for every $v \in V_T$ 
denoting the amount of flow from $r_T$ to $v$ and corresponding capacity 
variables $x_e$ denoting the amount of flow through $e$ for every $e \in E_T$.
The LP constraints guarantee the following:
\begin{enumerate}
	\item for each $i \in [q]$, $f$ supports a flow of at least $h_i$ from 
	the root $r$ to group $g'_i$,
	\item for each terminal $t \in S$, $f$ supports a flow of at most 1 to 
	the collection of its copies $M(t)$,
	\item the capacities given by $x$ support the flow $f$.
\end{enumerate}

It is not difficult to see that any integral feasible solution to this LP is 
an $r$-rooted tree which contains $h_i$ \emph{unique} terminals from each group $g_i$. 
While there are known algorithms for DCST on trees that are based on LP-rounding 
\cite{GuptaS06}, it is not clear if these techniques work for this modified LP. 
We describe a procedure that iteratively rounds solutions to our LP above by 
leveraging a connection to the minimum \emph{density} Directed Group Steiner 
Tree problem (MD-DGST) (defined in Section \ref{sec:prelim}). 

Let $\opt_{LP}$ denote the cost of a (fractional) optimal solution 
$(x^*, f^*)$ to the LP. Using our LP constraints and the fact that copies of 
distinct terminals are disjoint, we observe that a group $g_i$ can be 
partitioned into $h_i/2$ parts ($g_i = \uplus_{j \in [h_i/2]} g_i^{(j)}$) 
such that each part $g_i^{(j)}$ receives at least one unit of flow from 
the root $r_T$; here we assume that $h_i/2 \in \Z_+$ to simplify notation. 
We consider the MD-DGST instance defined on the tree $(\calT, c_T)$, 
root $r_T$, and groups $g'_{i, j} = \cup_{t \in g_i^{(j)}} M(t)$ for 
every $i \in [q]$ and $j \in [h_i/2]$. We note that since $(x^*, f^*)$ is 
feasible for our LP, it is a feasible fractional solution for a natural 
LP relaxation for MD-DGST (see \cite{ZosinK02}). Moreover, this fractional 
solution has density at most $2\opt_{LP}/k$, since the cost is $\opt_{LP}$ 
and the number of groups that receive at least one unit of flow is 
$\sum_{i \in [q]} h_i/2 = k/2$. To obtain a good feasible integral solution, 
we use the following result by Zosin and Khuller \cite{ZosinK02}.
\begin{theorem}[\cite{ZosinK02}]
  There exists a polynomial time $O(d)$-approximation (w.r.t the LP) for 
	the MD-DGST problem when the input graph is a directed out-tree 
	with height $d$.
\end{theorem}
Thus we can obtain a tree $T_1$ of density at most 
$O(d\ \opt_{LP}/k)$. We remove all terminals $t \in S$ such that 
$M(t) \cap T_1 \neq \emptyset$. We then repeat this process 
until we satisfy the requirement that 
every group $i$ contains at least $h_i$ distinct $t \in g_i$ such that 
$T' \cap M(t) \neq \emptyset$. This terminates in polynomial time since 
each iteration removes at least one terminal. 
Using standard inductive arguments, one can 
bound the total cost of the edges by $O(d \log k) \opt_{LP}$. Using 
Property \ref{thm:emb:height} of Theorem \ref{thm:planar_dir_tree_emb}, we can bound 
the height of $\calT$ by $O(\log N)$. Thus we obtain a tree $T'$ satisfying 
the desired conditions with cost $O(\log k \log N) \opt_{LP}$.
The rest of the proof is similar to that of DGST and so we omit the details for brevity.

%% file: planar_dst_lp.tex
\section{Integrality Gap of Cut-based LP Relaxation}
\label{sec:lp}

In this section, we prove Theorem \ref{thm:integrality-gap} on the integrality 
gap of the LP relaxation for DST on planar graphs, defined as follows. 
Let $\mathcal{C}:=\{U\subseteq V: r\in U \text{ and } U\cap S\neq S\}$. 

\begin{equation}\label{DST-LP}
\tag{DST-LP}
\begin{aligned}
	\min\quad \sum_{e\in E}c_ex_e& \\
	s.t.\quad \sum_{e\in \delta^+(U)} x_e&\geq 1 \quad \forall\ U \in \mathcal{C} \\
			x_e &\geq 0 \quad \forall\ e\in E
\end{aligned}
\end{equation}

This is a relaxation of the integer program where, for every edge $e\in E$, 
we have a variable $x_e\in \{0,1\}$ which indicates whether $e$ is contained 
in the solution.
For every subset $U$ of vertices containing root $r$ with $U\cap S \neq S$, 
a feasible solution must contain at least one edge in $\delta^+(U)$, since it 
must contain a path from the root $r$ to every terminal in $S \setminus U$.
The LP contains an exponential number of constraints but has an efficient 
separation oracle (an $s$-$t$ min-cut computation).
One can also formulate a compact extended formulation with additional flow 
variables (see \cite{LiL22}). 

We now prove that integrality gap of the LP is at most $O(\log^2 k)$ via a 
constructive procedure. One can view the algorithm as running the recursive 
algorithm of \cite{FriggstadM23}
by using the LP value as the estimate for the optimum. 
Given an arbitrary feasible solution $x$ of (\ref{DST-LP}), Algorithm 
\ref{algo:dir_plan_lp} constructs a directed Steiner tree.

\begin{algorithm}
  \caption{$\alglp(G, r, S, x)$}
  \label{algo:dir_plan_lp}
  \begin{algorithmic}
    \If{$|S| \le 6$} 
      \State \Return $T$ obtained by connecting each terminal $t$ to $r$ via 
			a shortest $r$-$t$ path
    \EndIf
    \State $\bar{x}_e \gets \left(1+\frac{1}{\log |S|}\right)\cdot x_e$ 
    \State $(P, C_1, \dots, C_\ell) \gets \algsep(G, r, S, 2\log |S| 
		\cdot \sum_{e\in E}c_e\cdot x_e)$ \hfill (Algorithm \ref{algo:prune_sep})
    \State $T \gets P$ 
    \For{$i = 1, \dots, \ell$}
      \State $\bar{x}^{(i)} \gets $ the values of $\bar{x}$ restricted to $E(C_i)$.
      \State Compute $\alglp(C_i, r, S \cap C_i, \bar{x}^{(i)})$ to obtain 
			tree $T_i = (V_i, E_i)$ with root $r_i$. 
      \State Augment $T$ with $T_i$ by replacing each edge from $r_i$ with an 
			edge from the corresponding uncontracted node of $P$.
    \EndFor \\
    \Return $(T, r)$
  \end{algorithmic}
\end{algorithm}

The base case is when there are at most six terminals; the algorithm connects each 
terminal to the root directly via shortest paths.
Otherwise, the algorithm first scales up the $x_e$ values for every $e\in E$ by a 
factor of $\left(1+\frac{1}{\log |S|}\right)$. Then, it calls 
$\algsep((G, c), r, S, 2 \log |S| \cdot \sum_{e\in E}c_e x_e)$ and obtains several 
sub-instances; recall that the separator $P$ is contracted into the root in each 
sub-instance.
The algorithm recursively solves these sub-instances and returns the corresponding 
solution in the original graph $G$. We observe that we do not re-solve the LP in 
the recursion but instead use the induced fractional solution after scaling up $x$.

We now analyze the correctness and cost of the tree returned by the algorithm with 
respect to the cost of the LP solution $\sum_e c_e x_e$. 
In the recursive case we first show that $\bar{x}$ is a feasible LP solution 
after removing the distant vertices. 

\begin{lemma}\label{lem:lp-feasible}
	Let $G = (V, E)$ be a directed planar graph with edge costs 
	$c: E \to \R_{\geq 0}$, a root $r \in V$, and a set of terminals 
	$S \subseteq V$ of size $k$. Let $x$ be a feasible solution of 
	(\ref{DST-LP}). Let $V':=\{v\in V: d(r,v)>2\log k\cdot \sum_{e\in E}c_ex_e\}$ 
	and $E'$ be set of edges that are incident to some vertex in $V'$. 
	Let $\bar{x}$ be a fractional solution where
	$\bar{x}_e:=0$ if $e\in E'$, and otherwise 
	$\bar{x}_e:=\left(1+\frac{1}{\log k}\right)\cdot x_e$. 
	Then, $\bar{x}$ is a feasible solution to (\ref{DST-LP}) for the given instance.
\end{lemma}

\begin{proof}
	We note that $\bar{x}_e\geq 0$ since $x_e\geq 0$ for every $e\in E$. It suffices to show that for every $U\in \mathcal{C}$, we have $\sum_{e\in \delta^+(U)} \bar{x}_e \geq 1$.
	Let $t\in S$ be a terminal such that $t\not\in U$. Since $x$ is a feasible solution of (\ref{DST-LP}), $x$ supports a unit flow from root $r$ to terminal $t$. Thus $\sum_e c_e x_e \ge d_G(r,t)$. Since every  path from root $r$ to a terminal $t$ passing through $V'$ has length more than $2\log k\cdot \sum_{e\in E}c_ex_e$ (by definition of $V'$), by Markov's inequality, the amount of flow supported by $x$ from $r$ to $t$ passing through $V'$ is smaller than $\frac{1}{2\log k}$. This tells us that after removing all vertices in $V'$, $x$ still supports a flow of at least $1-\frac{1}{2\log k}$ from $r$ to $t$. Hence, we have
	$
		\sum_{e\in \delta^+(U)} \bar{x}_e 
		= \left(1+\frac{1}{\log k}\right)  
		\cdot  
		\sum_{e\in \delta^+(U) \setminus E'} x_e 
		\geq \left(1+\frac{1}{\log k}\right) \cdot \left(1-\frac{1}{2\log k}\right) 
		\geq 1. 
	$
\end{proof}

The lemma below bounds the cost of the tree returned by 
Algorithm~\ref{algo:dir_plan_lp}, which concludes the proof of 
Theorem \ref{thm:integrality-gap}.

\begin{lemma}\label{lemma-lp}
  Given a directed planar graph $G = (V, E)$ with edge costs 
	$c: E \to \R_{\geq 0}$, a root $r \in V$, and a set of terminals 
	$S \subseteq V$ of size $k$. Let $x$ be a feasible solution of 
	(\ref{DST-LP}). $\alglp(G, r, S, x)$, returns a feasible directed Steiner 
	tree with cost $O(\log^2 k)\cdot \sum_{e\in E}c_ex_e$.
\end{lemma}

\begin{proof}
	By induction on the number of terminals, we prove that the cost of the tree 
	output by the algorithm is at most $6 (\log k+1)^2 \cdot  \sum_{e\in E}c_ex_e$.
	
	First we consider the base case when $k \le 6$. 
	We observe that if $x$ is a feasible solution then for every terminal $t$ 
	the length of the shortest $r$-$t$ path in $G$ is at most the fractional 
	LP cost. Thus the algorithm outputs a feasible solution whose cost is 
	at most $6 \sum_{e \in E} c_ex_e$.

	Consider the case where $k> 6$. Algorithm~\ref{algo:dir_plan_lp} finds 
	three paths $P_1, P_2, P_3$ and contracts their union into $r$ to 
	obtain graph $G_P$. The contraction creates several independent 
	sub-instances $(C_1, S\cap C_1), \ldots, (C_{\ell}, S\cap C_{\ell})$ 
	where $|S\cap C_i|\leq k/2$ for every $1\leq i \leq \ell$. 
	By Lemma \ref{lem:lp-feasible}, contraction preserves the feasibility 
	of the LP solution induced on the residual instance. Moreover, it is easy 
	to see that any integer solution to the residual instance together with 
	$P$ is a feasible integer solution to the original instance. 
	It remains to do the cost analysis. Let $COST(G, S)$ be the cost of the 
	tree output by Algorithm~\ref{algo:dir_plan_lp}. 

	The length of each of the paths $P_1,P_2,P_3$ is at most 
	$2\log k\cdot \sum_{e\in E}c_ex_e$, and thus their total cost is at most 
	$6\log k\cdot \sum_{e\in E}c_ex_e$. Since each edge $e$ of the graph $G-V'$ 
	is in at most one sub-instance, we have
	\begin{equation}
		\left(1+\frac{1}{\log k}\right)\cdot \sum_{e\in E}c_ex_e 
		\quad = \quad  
		\sum_{e\in E} c_e\bar{x}_e 
		\quad \geq \quad 
		\sum_{i=1}^{\ell}\sum_{e\in E(C_i)}c_e\bar{x}^{(i)}_e. \label{LP_analysis-1}
	\end{equation}
	
	By the induction hypothesis,
	\begin{align*}
		\sum_{i=1}^{\ell} COST(C_i', S\cap C_i) 
		&\leq \sum_{i=1}^{\ell} 6 (\log |S\cap C_i|+1)^2 \cdot \sum_{\mathclap{e\in E(C_i)}} c_e\bar{x}^{(i)}_e  \\
		&\leq \sum_{i=1}^{\ell} 6 \log^2 k \cdot \sum_{\mathclap{e\in E(C_i)}} c_e\bar{x}^{(i)}_e  & \text{(since $|C\cap C_i|\leq k/2$)}\\
		&\leq 6 \log^2 k \cdot \left(1+\frac{1}{\log k}\right) \cdot \sum_{e\in E}c_ex_e  & \text{(by inequality~\eqref{LP_analysis-1})}
	\end{align*}

	Thus,
	$\textnormal{COST}(G, S) \leq 6\log k\cdot \sum_{e\in E}c_ex_e+ 
	\sum_{i=1}^{\ell} \textnormal{COST}(C_i', S\cap C_i) 
	\le 6 (\log k+1)^2 \cdot \sum_{e\in E}c_ex_e$, completing the induction proof.
\end{proof}

%% file: multiroot.tex
\section{Multi-Rooted problems via Density Argument}
\label{sec:multi-root}
The main contribution of this section is the proof of Theorem 
\ref{thm:multiroot} regarding the multi-rooted versions of 
DST, DGST, DCST, and DPST. 
As we remarked, we cannot directly reduce the multi-rooted problems to the single-root
version while preserving planarity of the graph. We use a simple strategy via 
density-based arguments.
In Section \ref{sec:density} we discuss the min-density DST problem and other related 
versions; we apply this to the multirooted setting in Section \ref{sec:multiroot-final}.

\subsection{Min-Density DST and Extensions}
\label{sec:density}

In this section we describe a polynomial time approximation for MD-DST along with 
a few other simple variations of the DST problem involving constraints on number of 
terminals or prize-collecting versions. We define the problems as follows.

\mypara{Directed $\ell$-Steiner Tree ($\ell$-DST):} The input consists of 
$G=((V,E), c)$, a root $r\in V$, a set of terminals 
$S\subseteq  V \setminus \{r\}$, and an integer $1\le \ell \le |S|$. 
The goal is to find the minimum cost $r$-tree $T$ with $|T\cup S|=\ell$.

\mypara{Prize Collecting Directed Steiner Tree (PC-DST):} 
The input consists of $G=((V,E), c)$, a root $r\in V$, a non-negative integer 
weights $p$. The goal is to find the minimum cost $r$-tree where the cost of a 
set of edges of the $r$-tree $T$ is computed as 
$c(T) = \sum_{e\in T} c(e) + \sum_{v\not\in T} p(v)$.

\medskip

We remark that a polynomial time algorithm for MD-DST follows directly if we 
have such an algorithm for $\ell$-DST, since we can consider the optimal 
solutions to $\ell$-DST for all values of $\ell$ between $1$ and $|S|$ 
and output the solution with the minimum density ratio. 
This is also true of approximation algorithms, if we have an $\alpha$-approximation 
to $\ell$-DST, this gives an $\alpha$-approximation to MD-DST. 
Thus we will restrict our attention in this section to $\ell$-DST and PC-DST.
As mentioned in \Cref{rem:ext}, we can solve these problems by directly 
modifying the algorithm of \cite{FriggstadM23}. We will present an 
interpretation of this modification using the tree embedding we describe.

When the input graph $G$ is a tree rooted at $r$, there are folklore dynamic 
programming that solves both $\ell$-DST and PC-DST exactly. 
We give a brief overview of the dynamic programming solution here:

We may assume $G$ is a binary tree where non-leaf nodes have out degree at most 
two\footnote{If a node in $G$ has out-degree larger than $2$, we can create an out 
arrow to a new node with edge cost $0$ that handle all but one of its children. 
The resulting tree will still have size $O(n)$.}. 
For both problems, we will recursively compute the minimum cost solution for 
subtrees of $G$. For $\ell$-DST, we will actually recursively compute 
$\ell'$-DST for all values $\ell'<\ell$. Once we have recursively 
computed optimal solutions for the subtrees of the children of a node, 
it is straightforward to combine to get the optimal solution of a node.

When the input graph $G$ is a directed planar graph, we can use 
\Cref{thm:planar_dir_tree_emb} to construct a tree $\calT$.
However we cannot naively apply the dynamic programming algorithm as 
described above to $\calT$, 
since some terminals are duplicated. With a bit of care, we can still 
easily handle these duplications. In the construction of $\calT$, 
terminals are only duplicated when we make a recursive call corresponding 
to halving the guest value of the optimum solution $\gamma$. 
Thus in nodes that make this recursive call, we add a constraint that 
we are restricted to taking an optimum solution in only one of the 
two subtrees. It is easy to see that the proof of Property 
\ref{thm:emb:graph-to-tree} of \Cref{thm:planar_dir_tree_emb} 
satisfies this constraint; thus this framework only loses an 
$O(\log k)$ approximation factor in the cost.

\begin{theorem}\label{thm:density}
	There is an $O(\log k)$-approximation for MD-DST, $\ell$-DST, and PC-DST in 
	planar graphs that runs in polynomial time.
\end{theorem}

\subsection{Multi-rooted problems}
\label{sec:multiroot-final}

The following Lemma is an easy consequence of iteratively using the min-density 
algorithm and applying a standard inductive argument; we sketch the argument below. 

\begin{lemma}\label{lem:density-to-multi-root}
	Let $\mathcal{G}$ be a minor-closed family of graphs. Suppose there is an 
	$\alpha(k,n)$-approximation algorithm for the minimum-density DST (MD-DST) 
	problem on instances from graphs in $\mathcal{G}$ containing $n$ nodes and 
	$k$ terminals. Then there is an $O(\alpha(k,n) \log k)$-approximation for 
	the multi-root version of DST on graphs from $\mathcal{G}$ with $n$ nodes and 
	$k$ terminals.
\end{lemma}
\begin{proof}[Proof of Lemma \ref{lem:density-to-multi-root}](Sketch)
Consider an instance of multi-rooted DST on a graph $G \in \mathcal{G}$ 
with roots $r_1,r_2, \ldots, r_R$. We will assume that every terminal is 
reachable from at least one of the roots, otherwise there is no feasible 
solution. The algorithm is as follows. Let $F = \emptyset$. 
Use the min-density algorithm from each of the roots $r_i$ to compute an 
approximate solution $F_i$. Let $j$ be the root with the smallest density 
solution among $F_1,F_2,\ldots,F_R$. We add $F_j$ to $F$ and remove the 
terminals covered by $F_j$. We iterate this procedure until all terminals 
are removed. The algorithm terminates in polynomial-time since each 
iteration removes at least one terminal. We now argue about the approximation 
ratio of this algorithm. Let $F^*$ be an optimum minimal solution to 
the multi-rooted problem. It is easy to see that $F^*$ is a 
branching\footnote{A branching is a collection of edges in which each node has 
in-degree at most one.} with roots $r_1, r_2, \ldots, r_R$. 
Let $F_i^*$ be the tree rooted at $r_i$ and let $k_i$ be the number of 
terminals in $F_i^*$ (note that $k_i$ can be $0$). We have $\sum_i k_i = k$. 
Thus there is some $i'$ such that
$c(F^*_{i'})/k_{i'} \le  c(F^*)/k$. Via the approximation guarantee of the 
min-density algorithm, $c(F_j)/k_j \le \alpha(k,n) c(F^*)/k$. 
Via a simple and standard inductive argument in covering problems such as 
Set Cover, this implies that
the cost of all the edges added in the algorithm is at most 
$\alpha(k,n) (1 + \ln k) c(F^*)$.
\end{proof}

Thus, in order to approximately solve the multi-root version, it suffices to 
solve the min-density version of the single root problem. Combining 
Theorem \ref{thm:density} and Lemma \ref{lem:density-to-multi-root},
we obtain the following result:

\begin{corollary}
There is an $O(\log^2 k)$-approximation for multi-root DST in planar graphs.
\end{corollary}

\begin{remark}
Via density-based argument one can also prove that the integrality gap of a 
natural cut-based LP relaxation for the multi-root version of DST is at 
most $O(\log^3 k)$. This upper bound is unlikely to be tight; we leave the 
improvement in the bound to future work.
\end{remark}

The density-based argument extends in a natural fashion to multi-rooted 
versions of DGST, DCST, and DPST.
In fact, we are able to attain the approximation ratios that we get in 
Section~\ref{sec:applications} for the single-rooted case.
This is not surprising, as the algorithms for DGST, DCST, and DPST in trees 
can all be obtained through the corresponding min-density problems.
The tree embedding argument from Section~\ref{sec:emb} shows that one can 
reduce the min-density problem
in planar graphs to one on trees at the loss of an $O(\log N)$ factor in 
the approximation ratio. Moreover, the height of the resulting tree can be 
assumed to $O(\log N)$. For GST on trees with height $d$ there is an 
$O(d)$ approximation for the min-density problem \cite{ZosinK02}. 
Thus, there is an $O(\log^2 N)$-approximation for the min-density 
DGST in planar graphs. Combining the ingredients, we obtain an 
$O(\log k \log^2 N)$-approximation for the multi-root version of 
DGST in planar graphs. 
We obtain the same approximation factor for DCST in planar graphs, 
using a similar argument to that of Section \ref{sec:applications} 
to obtain an $O(\log^2N)$-approximation to the min-density DCST problem 
in planar graphs.
For DPST on trees, implicit in \cite{Calinescu_Zelikovsky_2005} is an 
algorithm that yields an 
$O(\frac{\log^{1+\eps}n}{\eps\log \log n})$-approximation for the 
min-density problem on trees. Combining it with the tree embedding and 
the iterative procedure, we obtain an 
$O(\frac{\log^{1+\eps}n\log k \log N}{\eps\log \log n})$-approximation 
for the multi-rooted version of DPST in planar graphs where $k = f(V)$.  
This proves Theorem~\ref{thm:multiroot}.

%% file: appendix_planar_embedding.tex
\section{Proof of Theorem \ref{thm:planar_dir_tree_emb} (Tree Embeddings for Directed Planar Graphs)}
\label{appendix:sec:emb}

Fix a directed planar graph $G = (V, E)$ with edge costs $c: E \to \R_{\geq 0}$, 
root $r \in V$ and terminals $S \subseteq V$, and let $\gamma \leq c(E)$.
Let $((\calT = (V_T, E_T), c_T), r_T, M) = \algemb((G, c), r, S, \gamma)$. 
In this section, we will show that 
$(\calT, r_T, M)$ satisfies the properties guaranteed by Theorem 
\ref{thm:planar_dir_tree_emb}. Following the theorem statement, we let 
$n := |V|, k := |S|$.

We denote by $(H, \phi)$ the root of the tree given by $\algemb(H, r, S \cap H, \phi)$; 
we omit the terminals and root from this label as they are implicit given $H$. 
For any $(H, \phi) \in V_T$,
we let $\calT_{H, \phi}$ denote the subtree of $\calT$
rooted at $(H, \phi)$, and let $S_H := S \cap H$ be the set of terminals in 
the sub-instance $H$, and let $k_H := |S_H|$. 
We denote by $(H, \phi)^*$ the auxiliary node $v^*$ 
created in $\algemb(H, r, S_H, \phi)$. 

\subsection{Size and Height Bound}

We prove Properties \ref{thm:emb:size} and \ref{thm:emb:height} 
via simple inductive arguments; these mostly follow from the analysis in 
\cite{FriggstadM23} and are rewritten here for completeness.

\begin{lemma}\label{lem:tree_emb_num_copies}
	$|M(t)| \leq k\gamma$ for all $t \in S$.
\end{lemma}
\begin{proof}
	Fix a terminal $t \in S$. For $(H, \phi) \in V_T$, we define $m_t(H, \phi)$ 
	to be the number of copies of $t$ in $\calT_{H, \phi}$; that is, $m_t(H, \phi) = 
	|M(t) \cap \calT_{H, \phi}|$. We will prove by induction on $\phi + k_H$ 
	that $m_t(H, \phi) \leq k_H \phi$.
	If $\phi < 1$ we return the empty tree, 
	so we start by considering $\phi = 1, k_H = 1$. Then 
	$\calT_{H, \phi}$ has two vertices $(H, \phi)$ and $t^P$ where $P$ is the 
	shortest $r$-$t$ path in $H$. Since $t^P \in M(t)$, $m_t(H, \phi) = 1$.
	
	Suppose $k_H \geq 2$. Let $(P, C_1, \dots, C_\ell) = \algsep(H, r, S_H, \phi)$,
	and let $k_i$ denote $|S_H \cap C_i|$ for $i \in [\ell]$. 
	We case on whether or not $t \in P$:
	
	\textbf{Case 1: $t \in P$}: Then, $t^P \in M(t) \cap \calT_{H, \phi}$, and
	$t \notin C_i$ for any $i \in [\ell]$, since $C_i \in G \setminus P$. Therefore, 
	$m_t(H, \phi) \leq 1 + m_t(H, \phi/2)$.
	
	\textbf{Case 2: $t \notin P$}: Since the $C_i$s are disjoint, there is exactly 
	one component $C_{j}$ that contains $t$. Thus $m_t(H, \phi) \leq 
	m_t(C_j', \phi) + m_t(H, \phi/2)$. 
	
	Combining the above cases and applying induction,
	\begin{align*}
			m_t(H, \phi, G) &\leq f_t(H, \phi/2) + \max(1, m_t(C_j, \phi)) 
			\leq k_H(\phi/2) + \max(1, k_{j}\phi) \\
			&\leq k_H(\phi/2) + (k_H/2)\phi \leq k_H\phi.
	\end{align*}
	The third inequality follows from the fact that $k_i \leq k_H/2$ for all 
	$C_i$ and that $k_H\phi/2 \geq 1$ when $\phi \geq 1, k_H \geq 2$. 
	This completes the inductive proof; thus $|M(t)| = m_t(G, \gamma) \leq k \gamma$.
	Since $t \in S$ was arbitrary, the claim holds for all terminals.
\end{proof}

\begin{lemma}
\label{lem:tree_emb_size}
	The number of vertices in $V_T$ is at most $O(k^3 \gamma)$.
\end{lemma}
\begin{proof}
	For $(H, \phi) \in V_T$, we define $f(H, \phi)$ be the number of vertices in the subtree 
	$T_{H, \phi}$ not including copies of terminals; that is, $f(H, \phi) = 
	|V(\calT_{H, \phi}) \setminus \cup_{t \in S} M(t)|$. 
	We will prove by induction on $\phi + k_H$ that $f(H, \phi) \leq k_H^3 \phi$.
	If $\phi < 1$ we return the empty tree, 
	so we start by considering $\phi = 1, k_H = 1$. Then 
	$\calT_{H, \phi}$ has two vertices $(H, \phi)$ and $t^P$ where $P$ is the 
	shortest $r$-$t$ path in $H$. Since $t^P \in M(t)$, $f(H, \phi) = 1$.
	
	Suppose $k_H \geq 2$. Let $(P, C_1, \dots, C_\ell) = \algsep(H, r, S_H, \phi)$,
	and let $k_i$ denote $|S_H \cap C_i|$ for $i \in [\ell]$. 
	Notice that $\calT_{H, \phi}$ consists of 
	$\calT_{H, \phi/2}$, $\calT_{C_i, \phi}$ for $i \in [\ell]$, 
	some nodes in $\cup_{t \in S} M(t)$, and two additional nodes $(H, \phi)$ and 
	$(H, \phi)^*$. By induction,
	\begin{align*}
			f(H, \phi) &= 2 + f(H, \phi/2) + \sum_{i \in [\ell]} f(C_i, \phi) 
			\leq 2 + k_H^3(\phi/2) + \sum_{i \in [\ell]} k_i^3 \phi \\
			&\leq 2 + k_H^3(\phi/2) + (k_H/2)^2\phi \sum_{i \in [\ell]} k_i 
			\leq 2 + k_H^3 \phi/2 + k_H^3\phi /4 \\
			&\leq k_H^3 \phi + \left[2 - \frac {k_H^3 \phi}4 \right]
			\leq k_H^3 \phi.
	\end{align*}
	The last inequality follows from the fact that $k_H \geq 2, \phi \geq 1$ implies 
	$\frac 1 4 k_H^3 \phi \geq 2$. This concludes the inductive proof.
	
	We obtain the bound on the number of vertices by in $\calT$ by combining the above with 
	Lemma \ref{lem:tree_emb_num_copies}, since $|V_T| = 
	|V_T \setminus \cup_{t \in S} M(t)| + |\cup_{t \in S} M(t)| = 
	f(G, \gamma) + \sum_{t \in S} |M(t)| \leq 
	k^3 \gamma + k^2 \gamma = O(k^3 \gamma)$.
\end{proof}

\begin{lemma}\label{lem:tree_emb_height}
	The height of $\calT$ is at most $O(\log(k\gamma))$.
\end{lemma}
\begin{proof}
	For $(H, \phi) \in V_T$, we define $h(H, \phi)$ be the height of the subtree 
	$T_{H, \phi}$. 
	We will prove by induction on $\phi + k_H$ that $h(H, \phi) \leq 2\log(k_H\phi) + 2$.
	If $\phi < 1$ we return the empty tree, and if $\phi = 1, k_H = 1$, then 
	$\calT_{H, \phi}$ has height 2.
	
	Let $(P, C_1, \dots, C_\ell) = \algsep(H, r, S_H, \phi)$,
	and let $k_i$ denote $|S_H \cap C_i|$ for $i \in [\ell]$. Notice that $h(H, \phi) 
	= 2 + \max(h(H, \phi/2), \max_{i \in [\ell]} h(C_i, \phi))$; the 
	2 comes from the initial path $(H, \phi), (H, \phi)^*$. By induction,
	\begin{align*}
			h(H, \phi) 
			&= 2 + \max(h(H, \phi/2), \max_{i \in [\ell]} h(C_i, \phi)) \\
			&\leq 2 + \max(2\log(k_H(\phi/2))+2, 2\log(k_i \phi) + 2) \\
			&\leq 2 + \max(2\log(k_H(\phi/2))+2, 2\log((k_H/2)\phi) + 2) \\
			&= 4 + 2\log(k_H \phi/2) = 4 + 2(\log(k_H\phi) - 1) = 2\log(k_H\phi) + 2.
	\end{align*}
	Thus the height of $\calT$ is $h(G, \gamma) \leq 2\log(k\gamma)+2$.
\end{proof}

\subsection{Projection from Graph to Tree}

We prove Property \ref{thm:emb:graph-to-tree} by showing that we can effectively 
embed a subtree $G' \subseteq G$ to a 
subtree of $\calT$ with the ``same'' terminal set. 
The idea is to follow the tree path from $r_T = (G, \gamma)$ to 
$(G, \phi)$, where $\phi$ is a close approximation to $c(G')$, and then buy the 
planar separator and recurse. The cost analysis in Lemma \ref{lem:proj_graph_to_tree} 
follows directly from the analysis in \cite{FriggstadM23}, rewritten here for completeness.

\begin{lemma}\label{lem:proj_graph_to_tree}
	For any $r$-tree $G' \subseteq G$ with $c(G') \leq \gamma$, there exists a 
	$r_T$-tree $T' \subseteq \calT$ with $c_T(T') = O(\log k) c(G')$, in which 
	for each terminal $t \in S \cap G'$, $M(t) \cap T' \neq \emptyset$. 
\end{lemma}
\begin{proof}
	We will prove the following by induction on $k_H$: 
	Let $(H, \phi_0) \in V_T$ and $H' \subseteq H$ be any 
	$r$-tree with $c(H') \leq \phi_0$. 
	Then, there exists a subtree 
	$T' \subseteq \calT_{H, \phi_0}$ rooted at $(H, \phi_0)$ 
	with $c_T(T') \leq 6(\log k_H + 1) c(H')$. Furthermore, for all 
	terminals $t \in S \cap H'$, $M(t) \cap T' \neq \emptyset$.
	
	If $k_H = 1$,
	$\calT_{H, \phi_0}$ has two vertices $(H, \phi_0)$ and $t^P$ connected 
	by an edge of cost $c(P)$, where $t$ is the unique terminal in $S_H$ and $P$ is the 
	shortest $r$-$t$ path in $H$. If $t \in H'$, then 
	$c_T(\calT_{H, \phi_0}) = c(P) \leq c(H')$, since $H'$ contains both $r$ and $t$.
	Furthermore, $t$ is the only terminal in $S \cap H'$, and 
	$t^P \in M(t) \cap \calT_{H, \phi_0}$. Thus $T' = \calT_{H, \phi_0}$ satisfies 
	the properties stated in the lemma. If $t \notin H'$ then $H'$ has no terminals, 
	so the singleton tree $\{r\}$ satisfies the stated properties.
	
	Assume $k_H \geq 2$. Let $\phi \leq \phi_0$ such that 
	$\phi/2 \leq c(H') \leq \phi$ and $(H, \phi) \in V_T$. Let $T'$ be the 
	subtree of $T_{H, \phi_0}$ consisting of the path from 
	$(H, \phi_0)$ to $(H, \phi)$ and the edge $((H, \phi), (H, \phi)^*)$.
	Let $(P, C_1, \dots, C_\ell) = \algsep(H, r, S_H, \phi)$,
	and let $k_i$ denote $|S_H \cap C_i|$ for $i \in [\ell]$. 
	Notice that for all $i \in [\ell]$, $k_i \leq k_H/2$ and $\phi 
	\geq c(H') \geq c(H' \cap C_i)$. Therefore, for each $i \in [\ell]$,
	we can apply induction on $(C_i, \phi)$ and $H' \cap C_i \subseteq C_i$ 
	to obtain a subtree $T_i' \subseteq \calT_{C_i, \phi}$ rooted at $(C_i, \phi)$. 
	We include each $T_i'$ in $T'$, along with the edges 
	$((H, \phi)^*, (C_i, \phi))$, $i \in [\ell]$. We also include in $T'$ all edges 
	$((H, \phi)^*, t^{P})$ for terminals $t \in P$.

	\textbf{Cost:} Note that $c(P) \leq 3\phi$; this is because we remove all vertices with 
	distance $\geq \phi$ and $P$ is the union of three shortest paths in the resulting graph.
	By choice of $\phi$, this implies that $c(P) \leq 6c(H')$. 
	Next, notice that $c_T(T') = c(P) + \sum_{i \in [\ell]} c_T(T_i')$, 
	since the only non-zero cost edge in $T' \setminus \cup_{i \in [\ell]} T_i'$ is 
	$((H, \phi), (H, \phi)^*)$, which has cost $c(P)$. By induction,
	\begin{align*}
	c_T(T') &= c(P) + \sum_{i \in [\ell]} c_T(T_i')
	\leq 6c(H') + \sum_{i \in [\ell]} 6(\log k_i+1)c(H' \cap C_i) \\
	&\leq 6c(H') + \sum_{i \in [\ell]}6 (\log k_H) c(H' \cap C_i) 
	\leq 6c(H') + 6 \log k_H c(H').
	\end{align*}
	Here, the second inequality follows from the fact that $C_i$ has at most 
	half the number of terminals of $H$, so $\log k_{i} \leq \log k_H - 1$,
	and the third inequality follows from the fact that all $C_i$s are disjoint.

	\textbf{Terminals:} Let $t \in S \cap H'$. Since $H' \subseteq H$ and 
	$H \subseteq \cup_{i = 1}^\ell C_i \cup P$, $t$ must either be 
	in $H' \cap C_i$ for some $i \in [\ell]$ or in $P$. If $t$ is in 
	$H' \cap C_i$, then by induction, $M(t) \cap T_i' \neq \emptyset$, so 
	$M(t) \cap T' \neq \emptyset$. Else, $t \in P$, in which case $t^{P} \in M(t) \cap T'$.

	This concludes the inductive argument and the proof of the lemma. 
\end{proof}

\subsection{Projection from Tree to Graph}

In this section we will prove Lemma \ref{lem:proj_tree_to_graph} which implies Property 
\ref{thm:emb:tree-to-graph}.

\begin{lemma}
\label{lem:proj_tree_to_graph}
For any $r_T$-tree $T' \subseteq \calT$ there exists a $r$-tree $G' \subseteq G$ 
with $c(G') \leq c_T(T')$. Furthermore, for each 
terminal $t \in S$, if $M(t) \cap T' \neq \emptyset$ then $t \in G'$. 
\end{lemma}

Fix some $r_T$-tree $T' \subseteq \calT$. Notice that there are only 
two types of non-zero cost edges in $\calT$; if $c_T(e) > 0$, then either 
\begin{enumerate}[(1)]
	\item $e = ((H, \phi), (H, \phi)^*)$ and $c_T(e) = c(P)$, where $P$ is the separator given 
	by \\$\algsep(H, r, S_H, \phi)$, or 
	\item $e = ((H, \phi), t^P)$ and $c_T(e) = c(P)$, where $t$ is the only terminal 
	in $S_H$ and $P$ is the shortest $r$-$t$ path in $H$. 
\end{enumerate}
Let $\calP(T')$ be the set of all such paths and separators corresponding to 
non-zero cost edges of $T'$. We let $G' = \cup_{P \in \calP(T')} P$.
It is clear that $G'$ can be computed efficiently by traversing $T'$ and 
including all relevant paths.

\begin{claim}
\label{claim:proj_cost}
$c(G') \leq c_T(T')$.
\end{claim}
\begin{proof}
	By definition of $\calP(T')$, $c_T(T') = \sum_{P \in \calP(T')} c(P)$. By construction
	of $G'$, $c(G') = c(\cup_{P \in \calP(T')} P) \leq \sum_{P \in \calP(T')} c(P)$.
\end{proof}

\begin{claim}
\label{claim:proj_terminals}
For each $t \in S$, if $M(t) \cap T' \neq \emptyset$ then $t \in G'$.
\end{claim}
\begin{proof}
	Fix $t \in S$ such that $M(t) \cap T' \neq \emptyset$. This means (by construction 
	of $M(t)$) that there exists some $t^P \in T'$, where $t \in P$ and $P$ is either a planar separator 
	or a shortest path. We will show that $P \in \calP(T')$, which implies $P \subseteq G'$. 
	
	First, suppose $P$ is a separator given by $\algsep(H, r, S_H, \phi)$ for some $(H, \phi) \in V_T$. 
	Then the edge $e^* = ((H, \phi), (H, \phi))$ must be in $T'$, since $T'$ is rooted 
	at $r_T$ and $e^*$ is on the unique tree path from $r_T$ to $t^P$. Since $e^* \in T'$, 
	$P \in \calP(T')$ so $P \subseteq G'$.

	Suppose instead that $P$ is a shortest path $r$-$t$ path computed in the base case of \\
	$\algemb(H, r_H, S_H, \phi)$ in which $S_H = \{t\}$. Then the edge 
	$e_t = ((H, \phi), t^P)$ must be in $T'$ since it is on the unique tree path from $r_T$ 
	to $t^P$, so once again $P \subseteq G'$.
\end{proof}

\begin{claim}
\label{claim:proj_subtree}
$G' \subseteq G$ is a directed out-tree rooted at $r$.
\end{claim}
\begin{proof}
	It suffices to show that $G'$ is connected and that all $v \in G'$ are reachable from $r$ 
	(if $G$ is cyclic, one can always remove edges 
	in cycles while maintaining the same set of terminals and ensuring that the cost does 
	not increase). For any $(H, \phi) \in V_T$, let $R_H \subseteq V$ be the subset of 
	vertices contracted into the root such that $H \subseteq G / (R_H \cup \{r\})$. We prove by 
	induction on the distance (i.e. number of edges) from the root $r_T$ that for all 
	$(H, \phi) \in V(T')$, all $v \in R_H$ are reachable from $r$ in $G'$. 

	Fix $(H, \phi) \in V(T')$. If the distance from $(H, \phi)$ to the root is 0, then 
	$(H, \phi) = (G, \gamma)$; $R_G = \{r\}$ which is clearly connected to $r$ in $G'$.
	Suppose the parent of $(H, \phi)$ in $\calT$ is $(H, 2\phi)$.
	Then since $T'$ is a rooted tree, $(H, 2\phi) \in V(T')$, so by induction, 
	all of $R_H$ is reachable from $r$ in $G'$.

	Otherwise, the parent of $(H, \phi)$ in $\calT$ must be some $(H', \phi)^*$. 
	Let $P$ be the planar separator given by $\algsep(H', r, S_{H'}, \phi)$. 
	This means that $H \subseteq H'/P$. Note that 
	$e^* = ((H', \phi), (H', \phi)^*) \in T'$, since it is on the unique path in $\calT$ 
	from $r_T$ to $(H, \phi)$. Thus $P \in \calP(T')$, so 
	$P \subseteq G'$. Note that $P$ is a directed outtree from 
	the root of $H'$, so all vertices $v \in P$ must be reachable 
	from the root of $H'$. 
	This root is the vertex resulting from contracting all nodes in $R_{H'}$;
	thus all $v \in P$ are reachable from $R_{H'}$ in $G'$.
	By induction, $R_{H'}$ is reachable 
	from $r$ in $G'$. Since $R_{H} = P \cup R_{H'}$, all vertices in $R_{H}$ must be 
	reachable from $r$ in $G'$, concluding the inductive proof. 

	To show that all of $G'$ is reachable from $r$, consider any $P \in \calP(T')$. 
	By construction, $P$ is either a shortest path or a planar separator computed in 
	$\algemb(H, r, S_H, \phi)$ for some $(H, \phi) \in T'$.
	In either case, $P$ is a directed 
	outtree from the root of $H$, so all of $P$ is reachable from $R_H$.
	By the above inductive proof, $R_H$ is reachable 
	from $r$ in $G'$, so $P$ must be as well.
\end{proof}